\documentclass{llncs}

\begin{document}
\title{First-order finite satisfiability  vs tree automata in safety verification}
\author{Alexei Lisitsa\inst{1}}
\institute{{Department of Computer Science, The University of Liverpool}
\email{{{A.Lisitsa@csc.liv.ac.uk}}}}
\bibliographystyle{plain}%
\maketitle

 
\begin{abstract}


In this paper we deal with verification of safety properties of term-rewriting systems. 
The verification problem is translated to  a purely logical problem of finding a finite 
countermodel for a first-order formula, which further resolved by a generic finite model finding procedure.
A finite countermodel produced during successful verification provides with a concise description of the system 
invariant sufficient to demonstrate  a specific safety property.     

We show the relative completeness of this approach with respect to the tree automata completion 
technique. On a set of examples taken from the literature we demonstrate  the efficiency of 
finite model finding approach as well as  its explanatory power.

 \end{abstract}



\section{Introduction}
The development of general automated methods for the 
verification of infinite-state and parameterized systems poses a major challenge. In general, such problems 
are undecidable, so one can not hope for the ultimate solution and the  development should focus on the restricted 
classes of systems and properties.   

In this paper we deal with a very general method for verification of \emph{safety} properties of infinite-state systems  which is based  on a  simple idea.  If an evolution of a computational system is faithfully modelled by a derivation in a classical first-order logic then safety verification (non-reachability of unsafe states)  can be reduced to the disproving of a first-order formula. The latter task can be  (partially, at least) tackled by generic automated procedures searching for \emph{finite} countermodels.  

Such an approach to verification was originated in the research on formal verification of security protocols \cite{Weid99,S01,GL08} and later has
 been extended to the wide classes of infinite-state and parameterised verification tasks. Completeness of the approach for  particular classes of systems (lossy channel systems) and 
relative completeness 
with respect to general method of regular model checking has been established in \cite{AL10} and \cite{AL10arxiv} respectively. 
 
Here we continue investigation of the boundaries of applicability of finite countermodels based method and are looking into verification of 
safety properties of  term-rewriting systems (TRS). Term-rewriting systems provide with a general 
formalism for specification and verification of infinite-state systems. 
Several general automated methods for verification of safety properties of term-rewriting systems has been proposed and 
implemented \cite{Timbuk2001,Reach2004,LPAR08} with the methods   based on tree automata completion \cite{Timbuk2001,Reach2004} playing the major role. 

We show that verification via finite countermodels (FCM) approach  provides with a viable alternative to the methods
based on the tree automata completion.  We show the relative completeness of FCM with respect to the tree automata completion methods (TAC). 

We illustrate it on a simple example taken from \cite{EqApp2010}.  Consider the TRS ${\cal R} = \{f(x) \rightarrow f(s(s(x))) \}$ and assume that
 we want to prove that $f(a) \not\rightarrow^{\ast} f(s(a))$. In \cite{EqApp2010} a simple finite-state abstraction of the set of reachable terms expressed by the equation $E = \{s(s(x) = x\}$  
is explicitly added to the TRS and simple analysis of rewriting modulo $E$ is proposed. In FCM approach, the same problem is translated into disproving of the first-order formula $\varphi_{\cal R} :=  (\forall x R(f(x),f(s(s(x)))) \rightarrow R(f(a),f(s(a))$. The intended meaning of the binary predicate $R$ here is to encode the reachability relation for the TRS.   
The finite countermodel of $\varphi_{\cal R}$, having the size 2 (cardinality of the domain) and  
essentially representing the above abstraction, i.e. satisfying $s(s(x) = x$, can be found by an automated model finder, e.g. Mace4 in a fraction of a second.

On a series examples taken from the literature we demonstrate practical efficiency of FCM approach using off-the shelf and state of the art implementation of a finite model finding procedure Mace4 (W. McCune); illustrate the high degree of automation achievable as well as the explanatory power of the method.

\section{Preliminaries}

In this paper we use standard terminology for first-order predicate logic and term-rewriting systems, and the  
for detailed accounts of these areas the  reader is referred to \cite{End} and to \cite{TR}, respectively. We remind here only the concepts which we are going to use in the paper.

\subsection{First-order Logic} 


The \emph{first-order vocabulary} is  defined as a finite set $\Sigma = {\cal F} \cup {\cal P}$ where ${\cal F}$ and ${\cal P}$ are the sets of functional and predicate symbols, 
respectively. Each symbol in $\Sigma$ has an associated arity, and we have ${\cal F} = \cup_{i \ge 0} {\cal F}{i}$ and  ${\cal P} = \cup_{i \ge 1} {\cal P}_{i}$, where ${\cal F}_{i}$ and
${\cal P}_{i}$ consist of symbols of arity $i$. The elements of ${\cal F}_{0}$ are also called \emph{constants}.

\emph{First-order model} over vocabulary $\Sigma$, or just a \emph{model} is  a pair ${\cal M} = \langle D, 
[\Sigma]_{D} \rangle $ where $D$ is a set called \emph{domain} of ${\cal M}$ and $[\Sigma_{D}]$ denotes the 
\emph{interpretations}  of all symbols from  $\Sigma$ in $D$. For a domain $D$ and a function symbol $f$ of arity $n \ge 1$ an interpretation of $f$ in $D$ is a function $[f]_{D}: D^{n} \rightarrow D$.  For a constant $c$ its interpretation $[c]_{D}$ is an element of $D$. For a domain $D$ and a predicate symbol $P$ of arity $n$ an interpretation of $P$ in $D$ is 
a relation of arity $n$ on $D$, that is $[P]_{D} \subseteq D^{n}$.  The model ${\cal M} = \langle D, 
[\Sigma]_{D} \rangle $ is called \emph{finite} if $D$ is a finite set. 

We assume that the reader is familiar with the standard definitions of first-order \emph{formula},  
first-order \emph{sentence}, satisfaction ${\cal M} \models \varphi$ of a formula $\varphi$ in a model ${\cal M}$, deducibility (derivability) $\Phi \vdash \varphi$ of a formula $\varphi$ from a set of formulae $\Phi$.   We also use the existence of 
\emph{complete} finite model finding procedures for the first-order predicate logic \cite{Model,McCune}, which given a first-order sentence $\varphi$ eventually produce a finite model for $\varphi$ if such a model exists.

\subsection{Term-rewriting systems and tree automata}

To define a term-rewriting system we fix a finite set of functional symbols 
${\cal F}$, each associated with an arity and a set of variables ${\cal X}$. 
${\cal T}({\cal F}, {\cal X})$ and ${\cal T} $ denote the set of terms and ground terms, respectively, defined in the standard way using ${\cal F}$ and ${\cal X}$.    
The set of variables of a term $t$ is denoted by $Var(t)$.  A substitution is a function $\sigma: {\cal X} \rightarrow {\cal T}({\cal F}, {\cal X})$, which can 
be extended homomorphically in a unique way (and keeping the name)  to $\sigma: {\cal T}({\cal F}, {\cal X}) \rightarrow {\cal T}({\cal F}, {\cal X})$.  Application 
of a substitution $\sigma$ to a term $t$ we denote by $t\sigma$.

A term-rewriting system ${\cal R}$ is a set 
of rewrite rules $l \rightarrow r$ where $l,r \in {\cal T}({\cal F}, {\cal X})$, $l \not\in {\cal X}$ and $Var(r) \subseteq Var(l)$. 
The notion of a \emph{subterm} is   defined in a standard way. \emph{One-step rewriting relation} $\Rightarrow_{\cal R} \subseteq  {\cal T}({\cal F}, {\cal X}) \times  {\cal T}({\cal F}, {\cal X})$ is defined as follows: $t_{1} \Rightarrow_{\cal R} t_{2}$ holds iff $t_{2}$ is obtained from $t_{1}$ by replacement of a 
subterm $l\sigma$ of $t_{1}$ with a subterm $r\sigma$ for some rewriting rule $(l \rightarrow r)$ in ${\cal R}$ and some substitution $\sigma$. 
The reflexive and transitive closure $\Rightarrow_{\cal R}$ is denoted by $\Rightarrow_{\cal R}^{\ast}$. 

Definitions of tree automata we borrow largely  from  \cite{}. Let $Q$ be a finite set of symbols called \emph{states} which we  formally treat as functional symbols of arity $0$ (constants). We assume $Q \cap {\cal F} = \emptyset$.  Elements of ${\cal T}({\cal F} \cup Q)$ are called \emph{configurations}.    

\begin{definition}(Transitions) 
A \emph{transition} is a rewrite rule $c \rightarrow q$, where $c$ is a configuration, i.e.  
$c \in {\cal T}({\cal F} \cup Q)$,  and $q \in Q$. A \emph{normalized} transition is a transition $c \rightarrow q$ where $c = f(q_{1}, \ldots, q_{n})$, $f$ is a functional symbol of arity $n$ from   ${\cal F}$, $q, q_{1}, \ldots q_{n} \in Q$. An $\epsilon$-\emph{transition} $c \rightarrow q$ is such that $c \in Q$.     
\end{definition} 

\begin{definition}(Tree automata)~\label{def:automata} 
A (bottom-up, non-deterministic, finite) tree automaton is a quadruple ${\cal A} = \langle F,Q,Q_{f},\Delta  \rangle$, where $Q_{f} \subseteq Q$ is a set of final (accepting) states and $\Delta$ is a set of 
normalized transitions and of $\epsilon$-transitions.   
\end{definition}

Transitions $\Delta$ of ${\cal A}$ induce the \emph{rewriting relation} on ${\cal T}({\cal F} \cup {\cal Q})$ which is denoted by $\Rightarrow_{\Delta}$ or $\Rightarrow_{\cal A}$.

\begin{definition} (Recognized language)~\label{def:language} 
The tree language recognized by ${\cal A}$ in a state $q$ is $L({\cal A}, q) = \{t \in {\cal T}({\cal F}) \mid t \Rightarrow_{\cal A}^{\ast} q\}$. The language recognized by ${\cal A}$ is ${\cal L}({\cal A}) = \cup_{q \in Q_{f}} {\cal L}({\cal A},q)$. 

\end{definition}

\begin{example} (Tree automaton and recognized language)
Let ${\cal F} = \{f,a,b\}$ and ${\cal A} = \langle {\cal F},Q,Q_{f}, \Delta \rangle$, where $Q = \{q_{1}, q_{2}\}$, $Q_{f} = \{ q_{1}\}$, and $\Delta = \{f(q_{1}) \rightarrow q_{1},  a \rightarrow q_{1}, b \rightarrow q_{2},q_{2} \rightarrow q_{1} \}$. Then ${\cal L}({\cal A},q_{1}) = {\cal T}(f,a,b)$, that is the set of all terms build on $\{f,a,b\}$,  
and ${\cal L}({\cal A},q_{2}) = \{b\}$. 
\end{example}

Deterministic bottom-up tree automata have the same expressive power as non-deterministic bottom-up tree automata, 
that is they recognize the same classes of term languages. In what follows we assume that automata are \emph{deterministic}, unless otherwise 
specified.

\section{Safety via finite countermodels}

\subsection{Basic verification problem} 

The main verification problem we consider in this paper is as follows.

\begin{problem}~\label{prob:basic}

\noindent
\begin{description}
\item[{\bf Given:}] Tree automata ${\cal A}_{I}$ and ${\cal A}_{U}$, a term-rewriting system 
${\cal R}$

\item[{\bf Question:}] Does  $\forall t_{1} \in {\cal L}({\cal A}_{I}) \; \forall t_{2} \in {\cal L}({\cal A}_{U})\; 
t_{1} \not\Rightarrow_{\cal R}^{\ast} t_{2}$ hold?    
\end{description}

\end{problem} 

In applications, we assume that the states of a  computational system to be verified are  
represented by terms, the system evolution (computation) is represented by ${\cal R}$;    
tree automata ${\cal A}_{I}$ and ${\cal A}_{U}$ provide with finitary specifications of the (infinite, in general) 
sets of allowed \emph{initial} states and the sets of \emph{unsafe} states, presented by 
${\cal L}({\cal A}_{I})$ and 
${\cal L}({\cal A}_{U})$, respectively.   Under such assumptions, safety of the system is equivalent to 
the positive answer on the question of the Problem~\ref{prob:basic}.   

Modifications of the basic problem will be considered later.

\subsection{Translation of the basic verification problem}~\label{sec:translation} 

In this subsection we show how to reduce the basic verification problem to the problem of disproving of a formula from classical first-order predicate logic. 

\noindent 
First,  we define the translation $\Phi_{R}$ of a term-rewriting system ${\cal R}$ over ${\cal T}({\cal F}, {\cal X})$ into a set of first-order formulae in the vocabulary ${\cal F} \cup \{R\}$, where $R$ is a new binary predicate symbol.     
Let $\Phi_{\cal R} =  \Phi^{r}_{\cal R} \cup \Phi_{\cal F}$, where 
$\Phi^{r}_{\cal R}  = \{R(l,r) \mid (l \rightarrow r) \in {\cal R}  \}$  
and $\Phi_{\cal F}$ is the set of the following formulae, which are all assumed to be universally closed and where 
$x_{1}, \ldots x_{i}, \ldots x_{n}, x_{i}^{'}$ are distinct variables: 

\begin{enumerate}
\item $R(x,y) \land R(y,z) \rightarrow R(x,z)$     \hspace*{37mm}  \emph{\bf transitivity axiom} 

\item  $R(x_{i}, x_{i}^{'}) \rightarrow R(f(x_{1}, \ldots, x_{i}, \ldots x_{n}),f(x_{1},\ldots, x_{i}^{'}, \ldots x_{n}))$  
for every $n$-ary functional symbol $f$ from ${\cal F}$ and every position $i$: $1 \le i \le n$  \\ \hspace*{80mm} \emph{\bf congruence axioms}

\end{enumerate}

Under such a translation first-order derivabiliy faithfully models rewriting in ${\cal R}$ as the following proposition shows.

\begin{proposition}~\label{prop:adequacyTR}
For ground terms $t_{1}$, $t_{2} \in {\cal T}({\cal F})$ if $t_{1} \Rightarrow^{\ast}_{\cal R} t_{2}$ then $\Phi_{\cal R} 
\vdash R(t_{1}, t_{2})$.  
\end{proposition} 

\begin{proof}
Due to the transitivity of $R$ specified in $\Phi_{\cal R}$ it is sufficient to show that if 
$t_{1} \Rightarrow_{\cal R} t_{2}$ then $\Phi_{\cal R}  \vdash R(t_{1}, t_{2})$.
Assume $t_{1} \Rightarrow t_{2}$ then $t_{2}$ is obtained from $t_{1}$ by the replacement of some subterm $l\sigma$ of 
$t_{1}$ with a subterm $r\sigma$ for some $(l \rightarrow r) \in {\cal R}$ and some substitution $\sigma$. Consider two sequences of subterms $\tau_{0} = l\sigma, \tau_{1}, \ldots, \tau_{k} = t_{1}$ and $\rho_{0} = r\sigma, \rho_{1}, 
\ldots, \rho_{k} = t_{2}$ with the property that $\tau_{i}$ is an immediate subterm of $\tau_{i+1}$ within $t_{1}$ and $\rho_{i}$ is an immediate  subterm of $\rho_{i+1}$ within $t_{2}$,  $i = 0, \ldots, k$.  Then we show by easy induction on $i$  that 
$\Phi_{\cal R} \vdash R(\tau_{i},\rho_{i})$ for $i = 0, \ldots, k$.  Indeed, for $i=0$ we have $R(\tau_{0},\rho_{0})\equiv R(l\sigma,r\sigma)$ is a ground instance of $R(l,r) \in  \Phi^{r}_{\cal R}$ and therefore $\Phi_{\cal R} \vdash R(\tau_{0},\rho_{0})$. For the step of induction, assume $\Phi_{\cal R} \vdash R(\tau_{i}, \rho_{i})$.  Notice that by construction of sequences of $\tau$'s and $\rho$'s $\tau_{i+1}$ and $\rho_{i+1}$ should have the same outermost functional symbol $f$ and coincide everywhere apart of subterms $\tau_{i}$ and $\rho_{i}$. Let $\tau_{i+1} = f(\ldots,\tau_{i},\ldots)$ and $\rho_{i+1} = f(\ldots,\rho_{i},\ldots)$. Then we have $R(\tau_{i},\rho_{i}) \rightarrow 
R(\tau_{i+1}, \rho_{i+1})$ is a ground instance of one of the formulae in $\Phi^{r}_{\cal R}$. So we have 
$\Phi_{\cal R} \vdash R(\tau_{i},\rho_{i}) \rightarrow 
R(\tau_{i+1}, \rho_{i+1})$ and by inductive assumption 
$\Phi_{\cal R} \vdash R(\tau_{i},\rho_{i})$. It follows $\Phi_{\cal R} \vdash R(\tau_{i+1},\rho_{i+1})$. The induction step is completed. We have $\Phi_{\cal R} \vdash R(\tau_{k},\rho_{k})$, which is $\Phi_{\cal R} \vdash R(t_{1}, t_{2})$

\end{proof}

Now we define a first-order translation of a tree automaton.

Let ${\cal A} = \langle {\cal F},Q_{I},Q_{f},\Delta  \rangle$ be a tree automaton. let $\Sigma_{\cal A}$ be the following first-order vocabulary: 

\begin{itemize}
\item constants for all elements of $Q$; 
\item all functional symbols from ${\cal F}$;  
\item a binary predicate symbol $R$; 
\end{itemize} 

Let $\Phi_{\cal A}$ to be the set of first-order formulae in vocabulary $\Sigma_{\cal A}$ defined as  
$\Phi_{\cal A} = \Phi_{\Delta} \cup \Phi_{\cal F}$, where $\Phi_{\Delta} = \{R(c,q) \mid (c \rightarrow q) \in \Delta \}$ 
and $\Phi_{\cal F}$ is as defined above.

As the following proposition shows first-order logic derivations from $\Phi_{\cal A}$ faithfully simulate the work of the automaton ${\cal A}$

\begin{proposition}~\label{prop:adequacyA} (Adequacy of automata translation)

If $t \in {\cal L}$ then  $\Phi_{\cal A} \vdash \vee_{q \in Q_{f}} R(t,q)$ 

\end{proposition}

\begin{proof}
The statement of the proposition follows immediately from 
Definitions~\ref{def:automata} and ~\ref{def:language} and Proposition~\ref{prop:adequacyTR}. 

\end{proof}

Now we are ready to define the translation of the basic verification problem. 
Assume we are given an instance $P = \langle{\cal A}_{I}, {\cal R}, {\cal A}_{U} \rangle $ of Problem~\ref{prob:basic}, 
with a term-rewriting system  ${\cal R}$ over ${\cal T}({\cal F}, {\cal X})$ and 
tree automata ${\cal A}_{I} = \langle {\cal F},Q_{I},Q_{{f}_{I}},\Delta_{I}  \rangle$, ${\cal A}_{U} = \langle {\cal F},Q_{U},Q_{{f}_{U}},\Delta_{U}  \rangle$. Assume also (without loss of generality) that sets ${\cal F}$, $Q_{I}$ and $Q_{U}$ are disjoint.

We define translation of $P$ as  $\Phi_{P}  = \Phi_{{\cal A}_{I}} \cup \Phi_{{\cal A}_{U}} \cup \Phi_{\cal R}$. By the above definitions we then also have $\Phi_{P} = \Phi_{\cal F} \cup \Phi_{\Delta_{I}} \cup \Phi_{\Delta_{U}} \cup \Phi^{r}_{\cal R}$.  We further define the translation of (negation of) correctness condition from $P$ as a formula 
$\psi_{P} = \exists x \exists y  \vee_{q_{i} \in Q_{I}, q_{u} \in Q_{U}} R(x,q_{i}) \land R(x,y) \land  R(y,q_{u})$. 

The following proposition and corollary serves as a formal underpinning of the proposed verification method. 

\begin{proposition}~\label{prop:correct}(Correctness of the translation)

Let $P$ be an instance of the basic verification problem as detailed above. 
Then if $P$  has a negative  answer then $\Phi_{P} \vdash \psi_{P}$
\end{proposition} 

\begin{proof}
The statement of the proposition immediately follows from Definitions~\ref{def:automata} and ~\ref{def:language} and Propositions~\ref{prop:adequacyTR} and ~\ref{prop:adequacyA}.
\end{proof} 

By contraposition we have the following 
                  
\begin{corollary}~\label{cor:correctness} 
If $\Phi_{P} \not\vdash \psi_{P}$ the instance $P$ has a positive answer and the safety property holds. 
\end{corollary}

\subsection{FCM method}

By FCM (finite countermodels) verification method we understand the following. Given an instance $P = \langle {\cal A}_{I}, {\cal R}, {\cal A}_{U}  \rangle$ of the basic verification problem, translate it into a set of first-order formulae $\Phi_{P}$ and a formula $\psi_{P}$ as described above. Then apply a generic finite model finding procedure to find a 
countermodel for  $\Phi_{P} \rightarrow  \psi_{P}$. If a countermodel found the safety property is established and the instance $P$ has got a positive answer.

\subsection{Relative completeness}

 In this section we show the \emph{relative completeness} of FCM with respect to verification methods based on tree automata completion techniques (TAC). More precisely, we show that if safety of TRS can be demonstrated by TAC, it can be demonstrated by FCM too. 

Given an instance $P$ of basic verification problem (Problem~\ref{prob:basic}) verification by TAC approach would proceed as follows. Starting from ${\cal A_{I}}$ and ${\cal R}$ completion procedure 
yields an automaton 
${\cal A}^{\ast}$ wich describes, in general, an \emph{overapproximation} of the set of terms reachable in ${\cal R}$ from ${\cal L}({\cal A_{I}}))$, that is 
$ {\cal L}({\cal A}^{\ast}) \supseteq \{t \mid  \exists t_{0} \in {\cal L}({\cal A_{I}}) \;\; t_{0} \rightarrow^{\ast}_{\cal R} t\}$.  Further, the check of whether ${\cal L}({\cal A^{\ast}}) \cap {\cal L}({\cal A_{U}}) = \emptyset$ is performed and, if it holds, the safety is 
established.

Exact description by of the set of all reachable terms  in a term-rewriting system by a tree automaton not always possible. The main direction in the development of TAC methods is a 
development of more efficient and more precise approximations methods.

\begin{theorem}~\label{th:completeness}
Let $P = \langle {\cal A}_{I}, {\cal A}_{U}, {\cal R} \rangle$ be a basic verification problem and there exists a tree automaton ${\cal A^{\ast}} = \langle {\cal F},Q^{\ast},Q^{\ast}_{f}, \Delta^{\ast}\rangle$ 
such that $ {\cal L}({\cal A}^{\ast}) \supseteq \{t \mid  \exists t_{0} \in {\cal L}({\cal A_{I}}) \;\; t_{0} \rightarrow^{\ast}_{\cal R} t\}$ and ${\cal L}({\cal A^{\ast}}) \cap {\cal L}({\cal A_{U}}) = \emptyset$. 
Then there exists a finite model ${\cal M}$ such that ${\cal M} \models \Phi_{P} \land \neg \psi_{P}$ (i.e ${\cal M}$ is a countermodel for 
$\Phi_{P} \rightarrow \psi_{P})$.

\end{theorem}

\begin{proof}
Assume the conditions of the theorem hold. 
Define the domain $D$ of the required model: $D = Q_{I}^{\bot} \times Q_{\ast}^{\bot} \times Q_{U}^{\bot}$, where $Q_{I}^{\bot} = Q_{I} \cup \{\bot \}$.

Define interpretations of contants $[c] = \langle a_{I}, a_{\ast}, a_{U} \rangle$, where 
\item $a_{x} = q$ if $(c,q) \in \Delta_{x}$, or $a_{x} = \bot$ otherwise, $x \in\{I, \ast, U \}$.

For a functional symbols $f$ of arity $n \ge 1$ define its interpretation 
$[f] : D^{n} \rightarrow D$ as follows  

$[f](\langle a_{I}^{1},a_{\ast}^{1},a_{U}^{1}\rangle, \ldots, \langle a_{I}^{n},a_{\ast}^{n},a_{U}^{n}\rangle) = \langle a_{I}, a_{\ast}, a_{U} \rangle$, where for all $x \in \{I, \ast, U\}$,  either $(f(a_{x}^{1}, a_{x}^{2}, \ldots a_{x}^{n}) \rightarrow a_{x}) \in \Delta_{x}$, or $a_{x} = \bot$, otherwise.

Once we defined the interpretations of all functional symbols (including constants) any ground term $t$ gets its interpretation $[t] \in D$ in a standard way. Then it  is an easy consequence of definitions that $[t]$ is a triple of states of automata ${\cal A}_{I}$, ${\cal A}_{\ast}$, ${\cal A}_{U}$, respectively,  into which they get working on the input $t$. More formally, if $[t] = \langle a_{I}, a_{\ast}, a_{U}\rangle $, then for all $x \in \{I, \ast, U\}$ either $t \Rightarrow_{x}^{\ast} a_{x} \in Q_{x}$, or there is no such $q \in Q_{x}$ that $t \Rightarrow_{x}^{\ast} q$, and 
then $t \Rightarrow_{x}^{\ast} a_{x} = \bot$.

Define the interpretation $[R] \subseteq D \times D$ of $R$ as follows. 

$$ [R] = \{\langle [t_{1}], [t_{2}] \rangle \mid t_{1}, t_{2} \mbox{ are  ground in } D, t_{1} \Rightarrow^{\ast} t_{2} \}$$  
 
where $\Rightarrow$ denotes $\Rightarrow_{\cal R} \cup \Rightarrow_{\Delta_{I}} \cup \Rightarrow_{\Delta_{U}}$

Now we are going to show that in a such defined model ${\cal M}$ we have $\Phi_{P} \land \neg\psi_{P}$ satisfied. 
Recall  $\Phi_{P} = \Phi_{\cal F} \cup \Phi_{\Delta_{I}} \cup \Phi_{\Delta_{U}} \cup \Phi^{r}_{\cal R}$. 

We have 

\begin{itemize}

\item ${\cal M} \models \Phi_{\cal F}$ (by definition of rewriting and definition of $[R]$) 
\item ${\cal M} \models  \Phi_{\Delta_{I}} \cup \Phi_{\Delta_{U}} \cup \Phi^{r}_{\cal R}$  (by definition of $[R]$)

\end{itemize}

To show ${\cal M} \models \neg \psi_{P}$ assume the opposite i.e ${\cal M} \models \psi_{P}$ that is \\
${\cal M}  \models  \exists x \exists y  \vee_{q_{i} \in Q_{I}, q_{u} \in Q_{U}} R(x,q_{i}) \land R(x,y) \land  R(y,q_{u})$. That means there are $a, b \in D$ such that $(a,[q_{i}]) \in [R]$, $(a,b) \in [R]$, $(b,[q_{u}]) \in [R]$. 
Consider the ground terms $\tau_{1}$ and $\tau_{2}$ such that $[\tau_{1}] = a$ and $[\tau_{2}] = b$.  We have 
$\tau_{1} \in {\cal L}({\cal A}_{I})$, $\tau_{1} \Rightarrow^{\ast} \tau_{2}$, $\tau_{2} \in {\cal L}({\cal A}_{U}$). It follows that $\tau_{2} \in {\cal L}({\cal A^{\ast}}) \cap {\cal L}({\cal A_{U}})$ which contradicts to the assumption of the theorem on emptiness of  ${\cal L}({\cal A^{\ast}}) \cap {\cal L}({\cal A_{U}})$. 

\end{proof}

\begin{note}
The above model construction serves only the purpose of proof and it is not efficient in practical use of the method. 
Instead we assume that the task of model construction is delegated to a generic finite model building procedure. 
\end{note}

\subsection{Variations on a theme}

Theorem~\ref{th:completeness} provides with a lower bound for  the verifying power of FCM method applied to a basic verification problem. 
In this section we consider practically important variations of the basic verification problem which allow  
simplified translations and  more efficient verification. 

\subsubsection{Finitely based sets of terms} 

In many cases of safety verification tasks for TRS the sets of initial and/or unsafe terms are given not by tree automata, but rather described as the 
sets of ground instances of terms from a given finite set of terms. More precisely, let $B$ be a finite set of terms in a vocabulary 
${\cal F}$ and  $g(B) = \{\tau \mid \exists t \in B  \land \tau = t\theta; \theta \mbox{ is ground }\}$. It is easy to see that for the finite $B$ 
$g(B)$ is a regular set. 

Consider the following modification of the basic verification problem. 

\begin{problem}~\label{prob:fb}

\noindent
\begin{description}
\item[{\bf Given:}] Finite sets of terms  $B_{I}$ and $B_{U}$, a term-rewriting system 
${\cal R}$ 

\item[{\bf Question:}] Does  $\forall t_{1} \in g(B_{I}) \; \forall t_{2} \in g(B_{U})\; 
t_{1} \not\Rightarrow_{\cal R}^{\ast} t_{2}$ hold?    
\end{description}

\end{problem} 

Let $P = \langle B_{I},{\cal R},B_{U} \rangle$ be an instance of the Problem~\ref{prob:fb}.  

The translation $\Phi_{\cal R}$ of the term rewriting system ${\cal R}$ is defined in \ref{sec:translation}.  
 
The translation of (negation of) correctness condition from $P$  is defined as   
$\psi_{P} = \exists \bar{x}   \vee_{t_{1} \in g(B_{I}),t_{2} \in g(B_{U}) } R(t_{1},t_{2})$. 

Now we have the following  analogue of Proposition~\ref{prop:correct}

\begin{proposition}~\label{prop:correct2}(Correctness of the translation)

Let $P$ be an instance of the basic verification problem as detailed above. 
Then if $P$  has a negative  answer then $\Phi_{\cal R} \vdash \psi_{P}$
\end{proposition}

\subsubsection{Rewriting strategies} 

Another simplification of the translation may come from the restrictions on the rewriting strategies in TRSs. 
If rewriting can only be applied at the outer level, i.e. redex can be only the whole term, not its proper subterm, then the first-order 
translation of an TRS can be simplified by using unary reachability predicate $R(-)$ instead of binary $R(-,-)$. 
The intended meaning of $R(t)$ is ``term $t$ is reachable from some of the initial terms (using outermost strategy)''.  We omit the obvious details of 
translation (axiomatization of $R$) and  rather refer to an Example~\ref{readers-writers}. 
Notice, that congruence axioms are not needed in this case and it was observed empirically that their absence makes the countermodel search 
more efficient. 


\section{Experiments}~\label{sec:experiments}
In this section we present three examples of application of FCM method for safety verification and compare the results with the results of alternative methods reported in the literature.

\subsection{Parity of $n^{2}$}

\begin{example}~\label{parity}

The following verification task is taken from \cite{Timbuk2001,LPAR08}.

Let $P_{n^{2}} = \langle {\cal A}_{I},{\cal R},{\cal A}_{U}\rangle$ be an instance of basic verification task. Term rewriting system ${\cal R}$ consists of the following rewriting rules 

\begin{itemize}
\item $plus(0,x) \rightarrow x$ 
\item $plus(s(x),y) \rightarrow  s(plus(x,y))$
\item $times(0,x) \rightarrow 0$
\item $times(s(x),y) \rightarrow plus(y,times(x,y))$
\item $square(x) \rightarrow times(x,x)$ 
\item $even(0)\rightarrow true$ 
\item $even(s(0)) \rightarrow false$
\item $even(s(x)) \rightarrow odd(x))$
\item $odd(0) \rightarrow false$
\item $odd(s(0)) \rightarrow true$ 
\item $odd(s(x)) \rightarrow even(x)$     
\item $even(square(x)) \rightarrow odd(square(s(x)))$ 
\item $odd(square(x)) \rightarrow even(square(s(x)))$ 
\end{itemize}

The tree automaton ${\cal A}_{I}$ recognizes the set of initial terms. It has the set of states $Q_{I} = \{s0,s1,s2\}$, the set
of the final states $Q_{I_{f}} = \{s0\}$ and the set of rewriting rules $\Delta_{I} = \{even(s1) \rightarrow s0, square(s2) \rightarrow s1, 
0 \rightarrow s2 \}$ It is easy to see that ${\cal L}({\cal A}_{I}) = \{even(square(0))\}$

The tree automaton ${\cal A}_{U}$ recognizes the set of unsafe terms. It has the set of states $Q_{U} = Q_{U_{f}} = \{q0\}$ and the set of rewriting rules 
$\Delta_{U} = \{false \rightarrow q0 \}$. 

So the question of the verification problem $P_{n^{2}}$ is whether $false$ is reachable from $even(square(0))$. 

First-order translation $\Phi_{P}$ of $P_{n^{n}}$ consists of the following formulae: 

\begin{itemize}

\item $R(plus(0,x),x)$ 
\item $R(plus(s(x),y), s(plus(x,y)))$
\item $R(times(0,x),0)$
\item $R(times(s(x),y),plus(y,times(x,y)))$
\item $R(square(x),times(x,x))$ 
\item $R(even(0),t)$
\item $R(even(s(0)),f)$
\item $R(even(s(x)),odd(x))$
\item $R(odd(0),f)$
\item $R(odd(s(0)),t)$
\item $R(odd(s(x)),even(x))$     
\item $R(even(square(x)),odd(square(s(x))))$ 
\item $R(odd(square(x)),even(square(s(x))))$ 
\item $R(x,y) \land R(y,z) \rightarrow  R(x,z)$ 
\item $R(x,y) \rightarrow R(even(x),even(y))$
\item $R(x,y) \rightarrow R(odd(x),odd(y))$
\item $R(x,y) \rightarrow  R(plus(x,z),plus(y,z))$ 
\item $R(x,y) \rightarrow  R(plus(z,x),plus(z,y))$ 
\item $R(x,y) \rightarrow  R(times(x,z),times(y,z))$ 
\item $R(x,y) \rightarrow  R(times(z,x),times(z,y))$ 
\item $R(x,y) \rightarrow  R(square(x),square(y))$
\item $R(0,s2)$
\item $R(even(s1),s0)$
\item $R(square(s2),s1)$ 
\item $R(f,q0)$

\end{itemize}


The formula $\psi_{P}: \exists x \exists y (R(x,s0) \land R(x,y) \land R(y,q0)$ expresses the negation of correctness condition. 

The finite model finder Mace4 has found a finite countermodel for $\Phi_{P} \rightarrow \psi_{P}$ 
(i.e  a finite model for $\Phi_{P} \land \neg \psi_{P}$) in 0.03s (see further details in \ref{subsec:exp}). The domain $D$ of the model  
is a two element set $\{0,1\}$. Interpretations of constants: [$f$] = [$q0$] = [$s1$] = [$s2$] = 0; [$s0$] = [$t$] = 1.  
Interpretations of functions: [$even$](0) = 1, [$even$](1) = 0;  [$odd$](0) = 0, [$odd$](1) = 1;
[$s$](0) = 1, [$s$](1) = 0; [$square$](0) = 0; [$square$](1) = 1; [$plus$](x,y) = $(x + y) mod 2$;
[$times$](x,y) = $x \times y$. Interpretation of reachability relation: [$R$] = $\{(0,0),(1,1)\}$. 

Notice that verification is done here automatically. This can be contrasted with the verification of the same system by  
a tree completion algorithm implemented in Timbuk system \cite{Reach2004},  where an user interaction was required to 
add an approximation equation $s(s(x)) = x$ manually. In \cite{LPAR08} an automated verification of the same 
system was reported    using Horn Clause approximation technique. The system was specified as a Horn Clause program and the 
verification followed by producing a model for the program which contained 53 elements. The above model produced by Mace4 within 
FCM approach provides with much more concise explanation of why the safety holds: interpretation of any ground term (0 or 1) 
is an invariant for reachability in TRS, [$even(square(0))$] = 1 and [$f$] = 0.        

\end{example}

\subsection{Readers-writers system verification}

In this subsection we consider the example of a readers-writers system verification  taken from \cite{Cla2007,EqApp2010}.

\begin{example}~\label{readers-writers}

In the  TRS specifying the system the only \emph{outermost} rewriting is possible, so for the translation we use \emph{monadic} 
reachability predicate. Furthermore, both the set of initial terms and the set of unsafe terms are finitely based. 
The vocabulary consists the constant $0$, unary functional symbol $s$ (for successor) and binary functional symbol $state$.

The rules are as follows

\begin{itemize}
\item $state(0,0) \rightarrow state(0,s(0))$
\item $state(x,0) \rightarrow state(s(x),0)$
\item $state(x,s(y)) \rightarrow state(x,y)$
\item $state(s(x),y) \rightarrow state(x,y)$
\end{itemize} 

The set of initial terms  is $I = \{state(0,0) \}$. 

The set of unsafe terms $U$ is finitely based with the base \\ $B = \{state(s(x),s(y)), state(x,s(s(y))) \}$.

The first-order translation $\Phi$ consists the conjunction of the following formulae

\begin{itemize}
\item $R(state(0,0))$
\item $R(state(0,0)) \rightarrow R(state(0,s(0)))$
\item $R(state(x,0)) \rightarrow R(state(s(x),0))$
\item $R(state(x,s(y))) \rightarrow R(state(x,y))$
\item $R(state(s(x),y)) \rightarrow R(state(x,y))$
\end{itemize}

The formula 
$\psi \equiv \exists x \exists y  R(s(x),s(y)) \lor R(x,s(s(y)))$ expresses the negation of the correctness condition.

The system can be then successfully verified by an FCM method. 
The search for the countermodel for $\Phi \rightarrow \psi$ took 0.01s  and  the model found is as follows. 

The domain $D$ of the model  
is a three element set $\{0,1,2\}$;  [$s$](0) = 1, [$s$](1) = 2, [$s$](2) = 2;  [$R$] = $\{(0,0),(0,1),(1,0),(2,0) \}$.

Notice that no additional information is needed for FCM method to automatically verify the  reader-writer system. That may be contrasted with 
 the verification using tree automata completion approach (Timbuk 3.0 system), reported in \cite{EqApp2010} where an equational abstraction 
rule $s(s(x)) = s(s(0))$ should be manually added to the TRS for the successful verification.

\end{example}  

\subsection{Reverse function}  
In this section we consider a verification problem from \cite{Timbuk2001}. The problem here is to show that list reverse function satisfies the 
following property: if in a list all symbols `a' are before all symbols `b' then after reversing there are no `a' before `b`.


\begin{example}~\label{reverse}

Vocabulary ${\cal F}$ consists of one 0-ary functional (constant) sumbol $0$ and three binary symbols $app$, $cons$, 
$rev$. 

The automaton recognizing is initial terms is defined as  ${\cal A}_{I} = \langle {\cal F}, Q_{I}, Q_{{f}_{I}}, 
\Delta_{I}  \rangle$, where ${\cal F}$ is as defined above; 
$Q_{I} = \{qrev,qlab,qlb,qa,qb\}$;  $Q_{f_{I}} = \{qrev\}$; $\Delta_{I}$ contains

\begin{itemize}
\item $rev(qlab) \rightarrow qrev$
\item $cons(qa,qlab) \rightarrow qlab$
\item $0 \rightarrow qlb$
\item $a \rightarrow qa$
\item $0 \rightarrow qlab$
\item $cons(qa,qlb) \rightarrow qlab$
\item $cons(qb,qlb) \rightarrow qlb$
\item $b \rightarrow qb$
\end{itemize}

The automaton recognizing \emph{unsafe} terms is defined as ${\cal A}_{U} = \langle {\cal F}, Q_{U}, Q_{{f}_{U}}, 
\Delta_{U}  \rangle$, where ${\cal F}$ is as above; $Q_{U} = \{qlab1,qlb1,q1,qa,qb\}$, $Q_{f_{U}} = \{ qlab1\}$; $\Delta_{U}$ contains 

\begin{itemize}
\item $cons(qa,qlab1) \rightarrow qlab1$
\item $cons(qa,qlb1) \rightarrow qlab1$
\item $cons(qa,q1) \rightarrow q1$
\item $a \rightarrow qa$
\item $0 \rightarrow q1$
\item $cons(qb,qlab1) \rightarrow qlab1$
\item $cons(qb,q1) \rightarrow qlb1$
\item $cons(qb,q1) \rightarrow q1$
\item $b \rightarrow qb$
\end{itemize}

The term-rewriting system ${\cal R}$ consists of the following rules

\begin{itemize}
\item $app(0,x) \rightarrow x$ 
\item $app(cons(x,y),z) \rightarrow cons(x,app(y,z))$
\item $rev(0) \rightarrow  0$
\item $rev(cons(x,y)) \rightarrow app(rev(y),cons(x,0))$
\end{itemize}
\end{example}

First-order translation $\Phi_{P}$ consists of the following formulae.  

\begin{itemize}
\item $R(rev(qlab),qrev)$ 
\item $R(cons(qa,qlab),qlab)$ 
\item $R(0,qlb)$
\item $R(a,qa)$
\item $R(0,qlab)$
\item $R(cons(qa,qlb),qlab)$
\item $R(cons(qb,qlb),qlb)$
\item $R(b,qb)$

\item $R(cons(qa,qlab1),qlab1)$
\item $R(cons(qa,qlb1),qlab1)$
\item $R(cons(qa,q1),q1)$
\item $R(0,q1)$
\item $R(cons(qb,qlab1),qlab1)$
\item $R(cons(qb,q1),qlb1)$
\item $R(cons(qb,q1),q1)$
\item $R(b,qb)$
\item $R(app(0,x),x)$
\item $R(app(cons(x,y),z),cons(x,app(y,z)))$
\item $R(rev(0),0)$
\item $R(rev(cons(x,y)),app(rev(y),cons(x,0)))$

\item $(R(x,y) \land R(y,z)) \rightarrow  R(x,z)$

\item $R(x,x)$ 
\item $R(x,y) \rightarrow  R(rev(x),rev(y))$ 
\item $R(x,y) \rightarrow  R(cons(z,x),cons(z,y))$ 
\item $R(x,y) \rightarrow R(cons(x,z),cons(y,z))$
\item $R(x,y) \rightarrow  R(app(z,x),app(z,y))$ 
\item $R(x,y) \rightarrow  R(app(x,z),app(y,z))$
\end{itemize}

The formula $\psi_{P}: \exists x \exists y ((R(rev(x),qrev) \land R(y,qlab1)) \land R(rev(x),y)$ expresses the negation of the correctness condition. 

For this standard encoding Mace4 has failed to find a countermodel for $\Phi_{P} \rightarrow \psi_{P}$ within 40000s. However after removing 
the congruence axiom  $R(x,y) \rightarrow R(rev(x),rev(y))$ Mace4 has found the model of size 3 (cardinality of the domain) in 0.06s. 
(see further details in  \cite{AL09}.  
The absence of such a congruence axiom means that no rewriting of proper subterms of $rev(\ldots)$  is allowed.
One can either easily argue that in TRS given above no such rewriting possible anyway, or, remaining in a pure automated verification scenario, 
just accept verification modulo restrictions on the rewriting strategy. This can be contrasted with the verification of the same system in 
\cite{Timbuk2001} using tree automata completion technique, which required  interactive approximation.

\subsection{Experimental results}~\label{subsec:exp}

\noindent 
In the experiments we used the finite model finder Mace4\cite{McCune} within the package 
Prover9-Mace4, Version 0.5, December 2007. 
It is not the latest available version, but it provides with convenient GUI for both the theorem prover and the finite 
model finder. The system configuration used in the experiments:   Microsoft Windows XP Professional, Version 2002, Intel(R) Core(TM)2 Duo CPU, 
T7100 @ 1.8Ghz  1.79Ghz,  1.00 GB of  RAM.  The time measurements are done by Mace4 itself, upon completion 
of the model search it communicates the CPU time used. The table below lists the parameterised/infinite state 
protocols together with the references and shows the time it took Mace4 to 
find a countermodel and verify a safety property.  
The time shown is an average of 10 attempts. $\infty$ means not return in 40000s.

\begin{center}
  \begin{tabular}{| l | c | r | }
    \hline
    Problem & Reference & Time \\ \hline
    Parity of $n^{2}$  & \cite{Timbuk2001} & 0.3s \\ \hline
    Readers-Writers & \cite{EqApp2010} & 0.01s \\
    \hline
    Reverse & \cite{Timbuk2001}  & $\infty$ \\
\hline
 Reverse (no congruence for rev) II &  Example~\ref{reverse} & 0.06s\\
\hline
  \end{tabular}
\end{center}

\section{Related work}

\subsection{Discussion and Related work}

The verification of safety properties  for term-rewriting systems using tree automata completion techniques has been addressed in 
\cite{Timbuk2001,Reach2004,EqApp2010}. The paper \cite{LPAR08} presents  a method based on encoding both term-rewriting system and tree automata into Horn logic and application of 
the static analysis techniques to compute a tree automaton accepting an approximation of the set of reachable terms.    
The main conceptual difference between these approaches and FCM presented in this paper, is that in \cite{Timbuk2001,Reach2004,EqApp2010,LPAR08} the safety verification is performed in two stages: first,  a tree automaton   approximating all reachable terms is 
obtained and it  depends only on TRS but \emph{not on the safety property}, and, second, an intersection of the language of this automaton with the language of unsafe states is computed. 
FCM method we presented here operates in  
one stage and computing \emph{regular approximations} (in terms of finite countermodels) is done for concrete safety properties. It has its disadvantage that the results of the verification of a TRS can not be re-used for the verification of different safety properties for the same TRS. On the other hand  this disadvantage is compensated by a higher degree of automation and higher explanatory power of FCM methods as our experimental results suggest.     Another advantage of FCM is its flexibility. Rewriting modulo theory can be easily incorporated into a general FCM framework and previous work on FCM  illustrates this point. In \cite{AL10} dealing with the verification 
of  lossy automata and cache coherence protocols, rewriting modulo first-order specifications of automata and modulo simple arithmetics, was used. In \cite{AL10arxiv} the translation of regular model  checking into FCM framework,  the  associativity of a monoid multiplication was explicitly specified.

As mentioned Section 1 the approach to verification using the modeling of protocol executions by 
first-order derivations and together with  countermodel finding for disproving was introduced within the research on the 
formal analysis of cryptographic protocols.  It can be traced back to the early papers by Weidenbach \cite{Weid99} and by 
Selinger \cite{S01}. In \cite{Weid99} a decidable fragment  of Horn clause logic has been identified for which resolution-based 
decision procedure has been proposed (disproving by the procedure amounts to the termination of saturation without producing a 
proof). It was also shown  that the fragment is expressive enough to encode cryptographic protocols 
and the approach has been illustrated by the automated verification of some protocols using the SPASS theorem prover. In \cite{S01}, 
apparently for the first time,  explicit building of finite countermodels has been proposed as a tool to establish correctness 
of cryptographic protocols. It has been illustrated by an example, where a countermodel was produced manually, and the
automation of the process has not been discussed.  The later work by Goubault-Larrecq \cite{GL08} has shown how a countermodel produced 
during the verification of cryptographic protocols can be converted into a formal induction  proof. 
Also, in \cite{GL08} different approaches to  model building have been discussed and it was argued that  
an implicit model building procedure using alternating tree automata  is more efficient in the situations when no small 
countermodels exist. Very recently, in the paper \cite{JW09} by J.~Jurgens and T.~Weber, an extension of Horn clause logic was proposed and 
the soundness of a countermodel finding procedure for this fragment has been shown, again in the context of cryptographic 
protocol verification. 

The work we reported in this paper differs from all the approaches mentioned previously in two important  aspects. 
Firstly, to the best of our knowledge,  none of the previous work addressed verification via countermodel finding applied 
outside of the  area of cryptographic protocols (that includes the most recent work \cite{Gut} we are aware of).
Secondly, the (relative)  completeness for the classes of verification tasks has not been addressed in  previous work.

\end{document}